\documentclass{article}
\usepackage[usenames,dvipsnames]{color}
\usepackage{amssymb,amsmath,eepic,graphics,color,enumerate,eucal}
\usepackage{tikz}
\usepackage{xspace}
\usetikzlibrary{calc}
\usepackage{todonotes}
\usepackage[utf8]{inputenc} 
\usepackage[T1]{fontenc}

\usepackage{epsfig}
\usepackage{amsmath,amsfonts,amssymb,epsfig,color,xy,amsthm}

\newtheorem{theorem}{Theorem}[]
\newtheorem{lemma}[theorem]{Lemma}
\newtheorem{corollary}[theorem]{Corollary}

\newtheorem{observation}[theorem]{Observation}

\newcommand{\myspan}{{\rm span}}

\newcommand{\ceil}[1]{\ensuremath{\left\lceil{#1}\right\rceil}}%
\newcommand{\ignore}[1]{}%

\newcommand{\ProblemFormat}[1]{{\sc #1}}
\newcommand{\ProblemName}[1]{\ProblemFormat{#1}\xspace}

\newcommand{\TCol}{\ProblemName{Generalized $T$-Coloring}}
\newcommand{\OrgTCol}{\ProblemName{$T$-Coloring}}
\newcommand{\GenLTCol}{\ProblemName{Generalized List $T$-coloring}}
\newcommand{\CA}{\ProblemName{Channel Assignment}}
\newcommand{\SetCover}{\ProblemName{SetCover}}
\newcommand{\ThreeSAT}{\ProblemName{3-CNF-SAT}}
\newcommand{\DomSet}{\ProblemName{Dominating Set}}

\newcommand{\heading}[1]{\medskip\noindent{\bf #1.\ }}%

\newcommand{\Aa}{{\ensuremath{\mathcal{A}}}}

\newcommand{\Cc}{{\ensuremath{\mathcal{C}}}}
\newcommand{\Dd}{{\ensuremath{\mathcal{D}}}}

\newcommand{\Ss}{{\ensuremath{\mathcal{S}}}}
\newcommand{\Ww}{{\ensuremath{\mathcal{W}}}}

\begin{document}
\title{Assigning channels\\ via the meet-in-the-middle approach\thanks{Research supported by National Science Centre of Poland, grant number UMO-2013/09/B/ST6/03136. An extended abstract was presented at 14th Scandinavian Symposium and Workshops on Algorithm Theory (SWAT 2014) in Copenhagen, Denmark.}}
\author{
  \L{}ukasz Kowalik\thanks{Institute of Informatics, University of Warsaw, Poland; email: \texttt{kowalik@mimuw.edu.pl}. }
  \and
  Arkadiusz Soca{\l}a\thanks{Faculty of Mathematics, Informatics and Mechanics, University of Warsaw, Poland; email: \texttt{arkadiusz.socala@students.mimuw.edu.pl}.}
}
\date{}

\maketitle

\begin{abstract}
We study the complexity of the \CA problem. 
By applying the meet-in-the-middle approach we get an algorithm for the $\ell$-bounded \CA (when the edge weights are bounded by $\ell$) running in time $O^*((2\sqrt{\ell+1})^n)$. This is the first algorithm which breaks the $(O(\ell))^n$ barrier.
We extend this algorithm to the counting variant, at the cost of slightly higher polynomial factor. 

A major open problem asks whether \CA admits a $O(c^n)$-time algorithm, for a constant $c$ independent of $\ell$.
We consider a similar question for \TCol, a CSP problem that generalizes \CA.
We show that \TCol does not admit a $2^{2^{o\left(\sqrt{n}\right)}} {\rm poly}(r)$-time algorithm, where $r$ is the size of the instance.
\end{abstract}

\section{Introduction}

In the \CA problem, we are given a symmetric weight function $w:V^2\rightarrow \mathbb{N}$ (we assume that $0\in\mathbb{N}$).
The elements of $V$ will be called vertices (as $w$ induces a graph on the vertex set $V$ with edges corresponding to positive values of $w$).
We say that $w$ is $\ell$-bounded when for every $x,y\in V$ we have $w(x,y)\le \ell$. 
An assignment $c:V\rightarrow\{1,\ldots,s\}$ is called {\em proper} when for each pair of vertices $x,y$ we have $|c(x)-c(y)|\ge w(x,y)$.
The number $s$ is called the {\em span} of $c$. The goal is to find a proper assignment of minimum span. 
Note that the special case when $w$ is $1$-bounded corresponds to the classical graph coloring problem.
It is therefore natural to associate the instance of the channel assignment problem with an edge-weighted graph $G=(V,E)$ where $E=\{uv\ :\ w(u,v)>0\}$ with edge weights $w_E:E\rightarrow \mathbb{N}$ such that $w_E(xy)=w(x,y)$ for every $xy\in E$ (in what follows we abuse the notation slightly and use the same letter $w$ for both the function defined on $V^2$ and $E$).
The minimum span is called also the span of $(G,w)$ and denoted by $\myspan(G,w)$.

It is interesting to realize the place \CA in a kind of hierarchy of constraint satisfaction problems. We have already seen that it is a generalization of the classical graph coloring.
It is also a special case of the constraint satisfaction problem (CSP). 
In CSP, we are given a vertex set $V$, a constraint set $\mathcal{C}$ and a number of colors $d$. 
Each constraint is a set of pairs of the form $(v, t)$ where $v\in V$ and $t\in \{1,\ldots,d\}$.
An assignment $c:V\rightarrow\{1,\ldots,d\}$ is {\em proper} if every constraint $A\in\mathcal{C}$ is satisfied, i.e.\ there exists $(v,t)\in A$ such that $c(v)\ne t$. 
The goal is to determine whether there is a proper assignment.
Note that \CA corresponds to CSP where $d=s$ and every edge $uv$ of weight $w(uv)$ in the instance of \CA corresponds to the set of constraints of the form $\{(u,t_1),(v,t_2)\}$ where $|t_1-t_2|<w(uv)$.

Since graph coloring is solvable in time $O^*(2^n)$~\cite{bhk:coloring} it is natural to ask  whether \CA is solvable in time $O^*(c^n)$, for some constant $c$. 
Unfortunately, the answer is unknown at the moment and the best algorithm known so far runs in $O^*(n!)$ time (see McDiarmid~\cite{mcdiarmid}). 
However, there has been some progress on the $\ell$-bounded variant.
McDiarmid~\cite{mcdiarmid} came up with an $O^*((2\ell+1)^n)$-time algorithm which has been next improved by Kral~\cite{kral} to $O^*((\ell+2)^n)$ and to $O^*((\ell+1)^n)$ by Cygan and Kowalik~\cite{cygan}. These are all dynamic programming (and hence exponential space) algorithms, and the last one applies the fast zeta transform to get a minor speed-up. 
Interestingly, all these works show also algorithms which {\em count} all proper assignments of span at most $s$ within the same running time (up to polynomial factors) as the decision algorithm.

It is a major open problem (see~\cite{kral,cygan,dagstuhl}) to find such a $O(c^n)$-time algorithm for $c$ independent of $\ell$ or prove that it does not exist under a reasonable complexity assumption. 
A complexity assumption commonly used in such cases is the Exponential Time Hypothesis (ETH), introduced by Impagliazzo and Paturi~\cite{IP01}.
It states that \ThreeSAT cannot be computed in time $2^{o(n)}$, where $n$ is the number of variables in the input formula.
The open problem mentioned above becomes even more interesting when we realize that under ETH, CSP does not have a $O^*(c^n)$-time algorithm for a constant $c$ independent of $d$, as proved by Traxler~\cite{traxler}.

\heading{Our Results}
Our main result is a new $O^*((2\sqrt{\ell+1})^n)$-time algorithm for the $\ell$-bounded \CA problem. Note that this is the first algorithm which breaks the $(O(\ell))^n$ barrier.
Our algorithm follows the meet-in-the-middle approach (see e.g. Horowitz and Sahni~\cite{HS-JACM74}) and is surprisingly simple, so we hope it can become a yet another clean illustration of this beautiful technique. We show also its (more technical) counting version, which runs within the same time (up to a polynomial factor).

Although we were not able to show that the unrestricted \CA does not admit a $O(c^n)$-time for a constant $c$ under, say ETH, we were able to shed some more light at this issue.
Let us consider some more problems in the CSP hierarchy.
In the \OrgTCol, introduced by Hale~\cite{Hale80}, we are given a graph $G=(V,E)$, a set $T\subseteq\mathbb{N}$, and a number $s\in\mathbb{N}$. 
An assignment $c:V\rightarrow\{1,\ldots,s\}$ is proper when for every edge $uv\in E$ we have $|c(u)-c(v)|\not\in T$. 
As usual, the goal is to determine whether there exists a proper assignment.
Like \CA, \OrgTCol is a special case of CSP and generalizes graph coloring, but it is incomparable with \CA. 
However, Fiala, Kr\'{a}l' and \v{S}krekovski introduced  which is a common generalization of vertex list-coloring (a variant of the classical graph coloring where each vertex has a list, i.e., a set of allowed colors), \CA and \OrgTCol.
The instance of the \GenLTCol is a triple $(G,\Lambda,t,s)$ where $G=(V,E)$ is a graph, $\Lambda : V\rightarrow 2^{\mathbb{N}}$, $t : E \rightarrow 2^{\mathbb{N}}$ and $s\in\mathbb{N}$, where $\mathbb{N}$ denotes the set of all nonnegative integers. 
An assignment $c:V\rightarrow\{1,\ldots,s\}$ is proper when for every $v\in V$ we have $c(v)\in\Lambda(v)$, and for every edge $uv\in E$ we have $|c(u)-c(v)|\not\in t(uv)$. 
As usual, the goal is to determine whether there exists a proper assignment.
Similarly as in the case of \CA, we say that the instance of \GenLTCol is $\ell$-bounded if 
$\max \bigcup_{e\in E} t(e)\le \ell$.
Very recently, the \GenLTCol was considered by Junosza-Szaniawski and Rzążewski~\cite{junosza-tcol}.  They show \GenLTCol can be solved in $O^*((\ell+2)^n)$ time, which matches the time complexity of the algorithm of Cygan and Kowalik~\cite{cygan} for \CA (note that an $\ell$-bounded instance of \CA can be seen as an $(\ell-1)$-bounded instance of \GenLTCol).
In this work we show that most likely one cannot hope for am $O^*(c^n)$-time algorithm for \GenLTCol.
We even consider a special case of \GenLTCol, i.e.\ the non-list version where every vertex is allowed to have any color, so the instance is just a triple $(G,t,s)$. We call it \TCol. We show that, under ETH, \TCol does not admit a $2^{2^{o\left(\sqrt{n}\right)}} {\rm poly}(r)$-time algorithm, where $r$ is the size of the instance (including all the bits needed to represent the sets $t(e)$ for all $e\in E$). Note that this rules out an $O(n!)$ algorithm as well.

\heading{Organization of the paper}
In Section~\ref{sec:dp} we describe an $O^*((\ell+2)^n)$-time dynamic programming algorithm for $\ell$-bounded \CA. It is then used as a subroutine in the $O^*((2\sqrt{\ell+1})^n)$-time algorithm described in Section~\ref{sec:mim}. In Section~\ref{sec:count} we extend the algorithm from Section~\ref{sec:mim} to counting proper assignments of given span.
Finally, in Section~\ref{sec:hardness} we discuss hardness of \TCol under ETH. 

\heading{Notation}
Throughout the paper $n$ denotes the number of the vertices of the graph under consideration.
For an integer $k$, by $[k]$ we denote the set $\{1, 2, \ldots, k\}$.
Finally, $\uplus$ is the disjoint sum of sets i.e. the standard sum of sets $\cup$
but with an additional assumption that the sets are disjoint.

\section{Yet another $O^*((\ell+2)^n)$-time dynamic programming}
\label{sec:dp}

In this section we provide a $O^*((\ell+2)^n)$-time dynamic programming algorithm for \CA.
It uses a different approach than e.g.\ the algorithm of Kral, and will be used as a subroutine in our faster algorithm.

For a subset $X\subseteq V$ and a function $f:X\rightarrow [\ell+1]$ let $\mathcal{A}_{X,f}$ be the set of all proper assignments $c:X\rightarrow\mathbb{N}$ of the graph $G[X]$ subject to the condition that for every $x\in X$ we have $c(x)\ge f(x)$. 

For every subset $X\subseteq V$ and $f:X\rightarrow [\ell+1]$ we compute the value of $T[X,f]$ which is equal to the minimum span of an assignment from $\mathcal{A}_{X,f}$. 
Clearly, the minimum span of $(G,w)$ equals to $T[V,f_1]$ where $f_1$ is the constant function which assigns $1$ to every vertex. 

The values of $T[X,f]$ are computed by dynamic programming as follows.
First we initialize $T[\emptyset, e_{[\ell+1]}]=0$ (where $e_{[\ell+1]}$ is the only function $f:\emptyset\rightarrow [\ell+1]$).
Next, we iterate over all non-empty subsets of $V$ in the order of nondecreasing cardinality.
In order to determine the value of $T[X,f]$ we use the recurrence relation formulated in the following lemma.

Informally, it uses the observation that there is a minimum-span assignment $c$ such that the vertex $v\in X$ with minimum color $c(v)$ is {\em left-shifted}, i.e.\ $c(v)=f(v)$. Hence we can check all possibilities for $v$ and then the colors of all the other vertices from $X$ have lower bounds in range $\{f(v),\ldots,f(v)+\ell\}$, so we can translate the range back down to $\{1,\ldots,\ell+1\}$ and use the previously computed values of $T[X\setminus\{v\},\cdot]$.

\begin{lemma}
For a subset $X\subseteq V$, a function $f:X\rightarrow [\ell+1]$ and a vertex $v$ define the function $f_v:X\setminus\{v\}\rightarrow [\ell+1]$ given by the formula \[f_v(x)=1+\max\{w(v,x),f(x)-f(v)\} \text{\quad\quad for every $x\in X\setminus\{v\}$}.\] 
Then,
\begin{equation}
\label{eq:dp}
T[X,f] = \min_{v\in X}(f(v) + T[X\setminus\{v\},f_v] - 1),
\end{equation}
\end{lemma}

\begin{proof}
Fix $v\in X$.
Denote $\Aa_{X,f,v} = \{c \in \Aa_{X,f}\ :\ c(v)=f(v)=\min_{x\in X} f(x)\}$.
Then, for every assignment $c\in\Aa_{X,f,v}$, for every $x\in X\setminus\{v\}$ we have $c(x)\ge f(v)+\max\{w(v,x),f(x)-f(v)\}$.
Hence, the minimum span of an assignment from $\Aa_{X,f,v}$ is equal to $f(v) + T[X\setminus\{v\},f_v]-1$.
It suffices to show that there is an assignment $c^*\in\mathcal{A}_{X,f}$ of minimum span such that $c^*(v)\in A_{X,f,v}$ for some $v\in X$. 
Consider an arbitrary assignment $c^*\in\mathcal{A}_{X,f}$ of minimum span. Let $x\in X$ be the vertex of minimum color, i.e.\ $c^*(x)$ is minimum.
If $c^*(x)=f(x)$ we are done. Otherwise consider a new assignment $c^{**}$ which is the same as $c^*$ everywhere except for $x$ and $c^{**}(x)=f(x)$; then $c^{**}$ is proper since $c^*(x)$ is minimal and clearly $c^{**}\in\mathcal{A}_{X,f}$. 
The span of $c^{**}$ is not greater than the span of $c^*$ (actually they are the same since $c^*$ has minimal span), so the claim follows. 
\end{proof}

The size of the array $T$ is $\sum_{i=0}^n{n\choose i}(\ell+1)^i=(\ell+2)^n$. Computing a single value based on previously computed values for smaller sets takes $O(n^2)$ time, hence the total computation takes $O((\ell+2)^nn^2)$ time. As described, it gives the minimum span only, but we can retrieve the corresponding assignment within the same running time using standard techniques.

\section{The meet-in-the-middle speed-up}
\label{sec:mim}

In this section we present our main result, an algorithm for $\ell$-bounded \CA that applies the meet-in-the-middle technique. Roughly, the idea is to find partial solutions for all possible {\em halves} of the vertex set and then merge the partial solutions efficiently to solve the full instance. 

For the clarity of the presentation we assume $n$ is even (otherwise we just add a dummy isolated vertex). 
Before we describe the algorithm let us introduce some notation.
For a set $X\subseteq V$, by $\overline{X}$ we denote $V\setminus X$. Moreover, for a function $f:X\rightarrow[\ell+1]$ we define function $\overline{f}:\overline{X}\rightarrow[\ell+1]$ such that for every $v\in\overline{X}$,
\[\overline{f}(v) = 1 + \max(\{1 + w(uv)-f(u)\ :\ uv\in E,\ u\in X\}\cup\{0\}).\]
The values $T[X,f]$ are defined as in Section~\ref{sec:dp}. Our algorithm is based on the following observation.

\begin{lemma}
\label{lem:mim}
The span of $(G,w)$ is equal to
\[\min (T[X,f] + T[\overline{X},\overline{f}] - 1),\]
where the minimum is over all pairs $(X,f)$ where $X\in{V\choose n/2}$ and $f:X\rightarrow[\ell+1]$.
\end{lemma}

\begin{proof}
Let $c^*:V\rightarrow\mathbb{N}$ be a proper assignment of minimum span $s$.
Order the vertices of $V=\{v_1,\ldots,v_n\}$ so that for every $i=1,\ldots,n-1$ we have $c^*(v_i)\le c^*(v_{i+1})$.
Consider the subset $X=\{v_1,\ldots,v_{n/2}\}$.
Let $s_1 = c^*(v_{n/2})$.
Define $f:X\rightarrow[\ell+1]$ such that $f(x)= 1 + \min\{s_1-c^*(x),\ell\}$ for every $x\in X$.
From the definition of $T$ we have $T[X,f]\le s_1$ (because the assignment $x\mapsto 1+s_1-c^*(x)$ belongs to $\Aa_{X,f}$ and has span $s_1$).
Moreover, note that for every $v\in\overline{X}$ it holds that 
\begin{equation*}
\begin{split}
c^*(v) & \ge \max(\{c^*(u)+w(uv)\ :\ uv\in E,\ u\in X\}\cup\{s_1\}) \\
       & =   \max(\{s_1 +w(uv) - f(u) + 1\ :\ uv\in E,\ u\in X\}\cup\{s_1\}) \\
       & =  s_1 - 1 + \overline{f}(v).
\end{split}
\end{equation*}
It follows that $s=\max_{v\in\overline{X}}c^*(v) \ge s_1 - 1 + T[\overline{X},\overline{f}] \ge T[X,f] + T[\overline{X},\overline{f}] - 1$.

Finally we show that $s > T[X,f] + T[\overline{X},\overline{f}] - 1$ contradicts the optimality of $c^*$. 
Let $c_1\in\Aa_{X,f}$ be an assignment of span $T[X,f]$ and let $c_2\in\Aa_{\overline{X},\overline{f}}$ be an assignment of span $T[\overline{X},\overline{f}]$. Consider the following assignment $c:V\rightarrow\mathbb{N}$.
\[
c(x) = \begin{cases}
            1 + T[X,f] - c_1(x) & \text{for $x \in X$} \\
            T[X,f] + c_2(x) - 1 & \text{for $x \in \overline{X}$}
           \end{cases} 
\]           
One can check that from the definition of $\overline{f}$ it follows that $c$ is a proper assignment. 
Moreover, the span of $c$ is equal to $T[X,f] + T[\overline{X},\overline{f}] - 1$. Hence, if $s > T[X,f] + T[\overline{X},\overline{f}] - 1$ then $c^*$ is not optimal, a contradiction.
\end{proof}

From Lemma~\ref{lem:mim} we immediately obtain the following algorithm for computing the span of $(G,w)$:

\begin{enumerate}
	\item Compute the values of $T[X,f]$ for all $X\in{V\choose \le n/2}$ and $f: X\rightarrow [\ell+1]$ using the algorithm from Section~\ref{sec:dp}.
	\item Find the span of $(G,w)$ using the formula from Lemma~\ref{lem:mim}.
\end{enumerate}

Note that Step 1 takes time proportional to $\sum_{i=0}^{n/2}{n \choose i}(\ell+1)^in^2=O(2^n(\ell+1)^{n/2}n^2)$. The size of array $T$ is clearly 
$O(2^n(\ell+1)^{n/2})$. In Step 2 we compute a minimum of ${n \choose n/2}(\ell+1)^{n/2}=O(2^n(\ell+1)^{n/2})$ values.
Hence the total time is $O(2^n(\ell+1)^{n/2}n^2)$. 
As described, the above algorithm gives the minimum span only, but we can retrieve the corresponding assignment within the same running time using standard techniques.
We have just proved the following theorem.

\begin{theorem}
For every $\ell$-bounded weight function the channel assignment problem can be solved in $O(2^n(\ell+1)^{n/2}n^2)$ time.
\end{theorem}

\section{An Extension to Counting}
\label{sec:count}

In this section we present an extension of our meet-in-the-middle algorithm which finds the number of proper assignments of span $s$.
This is slightly more technical than the decision algorithm because we need to avoid counting the same assignment more than once. We assume here that $V=\{1,\ldots,n\}$ (we will use the fact that $V$ is linearly ordered).

For $X\in {V\choose n/2}$, function $f:X\rightarrow [\ell+1]$ and value $r=1,\ldots,s$ denote the set of all assignments from $\mathcal{A}_{X,f}$ with span $r$ by $\mathcal{A}_{X,f,r}$.
Let us denote $Q[X,f,r]=|\mathcal{A}_{X,f,r}|$. We will use the recurrence relation formulated in the following lemma.

\begin{lemma}
\label{lem:eq-dp2}
For a subset $X\in {V\choose \le n/2}$, a function $f:X\rightarrow [\ell+1]$ and a vertex $v$ define the function $f_v:X\setminus\{v\}\rightarrow [\ell+1]$ given by the formula 
\[f_v(x)=\max\{f(x),1+w(vx),1+[x<v]\} \text{\quad\quad for every $x\in X\setminus\{v\}$}.\] 
Also, for a function $f:X\rightarrow \{2,\ldots,\ell+1\}$ define the function $f_{\downarrow}:X\rightarrow [\ell+1]$ given by the formula 
\[f_{\downarrow}(x)=\max\{f(x)-1,1\} \text{\quad\quad for every $x\in X$}.\] 

Then, for every $X\in {V\choose n/2}$, $f:X\rightarrow [\ell+1]$ and $r=1,\ldots,s$
\begin{equation}
\label{eq:dp2}
Q[X,f,r] = \begin{cases}
            \sum_{v\in f^{-1}(1)}Q[X\setminus\{v\},f_v,r] + [r>1]Q[X,f_{\downarrow},r-1] & \text{if $X\ne\emptyset$} \\
            [r=1] & \text{otherwise}
           \end{cases} 
\end{equation}
\end{lemma}

\begin{proof}
The proof is by induction on $|X|+r$.
The formula~\eqref{eq:dp2} clearly holds when $X=\emptyset$, since there is exactly one assignment with empty domain, it is proper and its span is 1.

Assume $X\ne\emptyset$.
The set $\mathcal{A}_{X,f,r}$ partitions into two subsets $\mathcal{B}$ and $\mathcal{C}$, where $\mathcal{B}$ contains the assignments which assign color 1 to some vertex and $\mathcal{C}$ contains the remaining assignments. 

We can further partition $\mathcal{B}=\bigcup_{v\in f^{-1}(1)}\mathcal{B}_v$, where 
\[\mathcal{B}_v = \{c\in\mathcal{B}\ :\ \min c^{-1}(1)=v\}.\]
Define $\mathcal{B}'_v=\{c|_{X\setminus\{v\}}\ :\  c \in \mathcal{B}_v\}$.
Then $|\mathcal{B}'_v|=|\mathcal{B}_v|$. 
Consider an arbitrary $c\in \mathcal{B}_v$. 
Then for every $x\in X\setminus\{v\}$ we have $c(x)\ge f(x)$, $c(x)\ge f(v)+w(vx)=1+w(vx)$, and if $x<v$ then $c(x)\ge 2$.
In other words, for every $x\in X\setminus\{v\}$ we have $c(x)\ge f_v(x)$ and hence $c|_{X\setminus\{v\}} \in \mathcal{A}_{X\setminus\{v\},f_v,r}$.
It follows that $\mathcal{B}'_v \subseteq \mathcal{A}_{X\setminus\{v\},f_v,r}$.
It is also easy to verify that every assignment $c'\in\mathcal{A}_{X\setminus\{v\},f_v,r}$ can be extended to a proper assignment $c\in\mathcal{B}_v$ by putting $c(v)=1$ and $c|_{X\setminus\{v\}}=c'$. 
Hence  $\mathcal{A}_{X\setminus\{v\},f_v,r} \subseteq \mathcal{B}'_v$.
It follows that $\mathcal{B}'_v=\mathcal{A}_{X\setminus\{v\},f_v,r}$ and hence $|\mathcal{B}_v|=|\mathcal{A}_{X\setminus\{v\},f_v,r}|=Q[X\setminus\{v\},f_v,r]$, where the last equality follows from the induction hypothesis. We get $|\mathcal{B}|= \sum_{v\in f^{-1}(1)}Q[X\setminus\{v\},f_v,r]$.

If $r=1$ then $\Cc=\emptyset$. Assume $r>1$.
It is clear that the assignments in $\mathcal{C}$ are in 1-1 correspondence with the assignments in $\mathcal{C}'=\{c_{\downarrow}\ :\ c\in\mathcal{C}\}$ and the assignments in $\mathcal{C}'$ have span $r-1$. Hence $|\mathcal{C}|=|\mathcal{A}_{X,f_{\downarrow},r-1}|= Q[X,f_{\downarrow},r-1]$, where the last equality follows from the induction hypothesis.

To sum up, \[|\mathcal{A}_{X,f,r}|=|\mathcal{B}|+|\mathcal{C}|=\sum_{v\in f^{-1}(1)}Q[X\setminus\{v\},f_v,r] + [r>1]\cdot Q[X,f_{\downarrow},r-1],\] as required.
\end{proof}

With Lemma~\ref{lem:eq-dp2} it is easy to describe a dynamic programming algorithm which for every subset $X\in {V\choose n/2}$, function $f:X\rightarrow [\ell+1]$ and value $r=1,\ldots,s$ computes the value of $Q[X,f,r]$. 
First we initialize $Q[\emptyset, \emptyset,r]=[r=1]$ for every $r=1,\ldots,s$ and next the values of $Q[X,f,r]$ are computed according to Formula~\eqref{eq:dp2}, using previously computed values of array $Q$; to this end we iterate over the triples $(X,f,r)$ in nondecreasing order of $|X|+r$. 
The number of triples considered is $O(2^n(\ell+1)^{n/2} s)$ and processing each triple takes $O(n^2)$ time. We have just shown the following.

\begin{lemma}
\label{lem:compute-Q}
There is an $O(2^n(\ell+1)^{n/2} sn^2)$-time $O(2^n(\ell+1)^{n/2} s)$-space algorithm which finds the values of $Q[X,f,r]$ for all subsets $X\in {V\choose n/2}$, functions $f:X\rightarrow [\ell+1]$ and values $r=1,\ldots,s$. 
\end{lemma}

If we use just the values of $Q[X,f,r]$ in the merge phase of the meet-in-the-middle approach, it is unclear how to avoid double-counting the same assignments. 
To overcome this problem, for a subset $X\subseteq V$, a function $f:X\rightarrow [\ell+1]$ and a  value $r=1,\ldots,s$ define $\mathcal{A^*}_{X,f,r}$ as the set of all proper assignments $c:X\rightarrow\mathbb{N}$ of the graph $G[X]$ such that $c$ has span $r$ and for every $x\in X$ , if $f(x)\le \ell$ then $c(x)=f(x)$ and otherwise $c(x)\ge f(x)$. Denote $Q^*[X,f,r]=\mathcal{A^*}_{X,f,r}$.
Observe the following.

\begin{observation}
\label{obs-Q*}
For a subset $X\subseteq V$ and a function $f:X\rightarrow [\ell+1]$ define the function $f_{\leftarrow\ell}:X\setminus f^{-1}([\ell])\rightarrow [\ell+1]$ given by the formula 
\[f_{\leftarrow\ell}(x)=\max(\{f(y) + w(yx) - \ell\ :\ {y\in f^{-1}([\ell])}\}\cup\{1\}),\] 
for every $x\in X\setminus f^{-1}([\ell])$. Then, for every $X\subseteq V$, $f:X\rightarrow [\ell+1]$ and $r=1,\ldots,s$ 
\begin{enumerate}[$(i)$]
\item if $f|_{f^{-1}([\ell])}$ is not a proper assignment then $Q^*[X,f,r]=0$;
\item if $f|_{f^{-1}([\ell])}$ is a proper assignment and $r\le \ell$ then 
\[Q^*[X,f,r]=[f^{-1}(\{r\})\ne\emptyset\text{ and }f^{-1}(\{r+1,\ldots,\ell+1\})=\emptyset];\]
\item if $f|_{f^{-1}([\ell])}$ is a proper assignment and $r\ge \ell+1$ then 
\begin{equation}
\label{eq:dp3}
Q^*[X,f,r] = Q[X\setminus f^{-1}([\ell]),f_{\leftarrow\ell},r-\ell]. 
\end{equation}
\end{enumerate}
\end{observation}

Now we proceed to the merge phase of our meet-in-the-middle algorithm.
For a function $f:X\rightarrow[\ell+1]$ we define function $\tilde{f}:\overline{X}\rightarrow[\ell+1]$ such that for every $v\in\overline{X}$,
\[\tilde{f}(v) = 1 + \max(\{1+w(uv)-f(u)\ :\ uv\in E,\ u\in X\}\cup\{[v<\max f^{-1}(1)]\}).\]
The role of the function $\tilde{f}$ is similar as $\overline{f}$ in determining the span using the meet-in-the-middle approach; the only difference is that if for some $x\in X$ we have $f(x)=1$ then for every $v\in\overline{X}$, if $v<x$ then $\tilde{f}(v)\ge 2$.
Informally, this helps us to avoid counting the same assignment once for every partition of the ``middle color'' into parts of relevant sizes.
Now we can formulate the counting counterpart of Lemma~\ref{lem:mim}. 

\begin{lemma}
\label{lem:counting-mim}
For a given graph $G$, weight function $w$ and integer $s\in\mathbb{N}$ the number of 
proper assignments of span $s$ is equal to 
\[\sum_{s^*=1}^s\sum_{X\in {V\choose n/2}}\sum_{\substack{f:X\rightarrow[\ell+1]\\f^{-1}(1)\ne\emptyset}} Q^*[X,f,s^*]\cdot Q[\overline{X},\tilde{f},s-s^*+1].\]
\end{lemma}

\begin{proof}
Let $\mathcal{D}$ be the set of all proper assignments of span $s$.
For an assignment $c\in\mathcal{D}$ define a total order of $V$ as follows: for $i,j\in V$ we have $i\prec j$ iff $(c(i),i)\le_{{\rm lex}} (c(j),j)$, where $\le_{{\rm lex}}$ is the lexicographic order. Then $c$ defines a permutation of the vertices $v^c_1\prec v^c_2 \prec \cdots \prec v^c_n$. Then $\mathcal{D}=\biguplus_{s^*=1}^s\mathcal{D}_{s^*}$, where
\[\mathcal{D}_{s^*} = \{c\in\mathcal{D}\ :\ c(v^c_{n/2})=s^*\}\]
Moreover, $\Dd_{s^*}=\biguplus_{X\in {V\choose n/2}}\Dd_{s^*,X}$, where
\[\Dd_{s^*,X} = \{c\in\Dd_{s^*}\ :\ \{v^c_1,\ldots,v^c_{n/2}\}=X\}.\]
Finally, \[\Dd_{s^*,X}=\biguplus_{\substack{f:X\rightarrow[\ell+1]\\ f^{-1}(1)\ne\emptyset}}\Dd_{s^*,X,f},\] 
where $\Dd_{s^*,X,f}$ is the set of assignments $c\in\Dd_{s^*,X}$ such that
for every $x\in X$, if $f(x)\le\ell$ then $c(x)=s^*-f(x)+1$ and if $f(x)=\ell+1$ then $c(x)\le s^*-f(x)+1$. Note that the condition $f^{-1}(1)\ne\emptyset$ is necessary to satisfy the defining condition of $\Dd_{s^*}$; in particular $v^c_{n/2}=\max f^{-1}(1)$.

Consider an arbitrary $c\in \Dd_{s^*,X,f}$.
Now observe that for every $v\in\overline{X}$ and $u\in X$ such that $uv\in E$, we have
$c(v)\ge \max\{c(u)+w(uv),s^*\}$. Moreover, if $v<v^c_{n/2}$, i.e.\ $v<\max f^{-1}(1)$, then $c(v)\ge s^*+1$. Hence,
\begin{equation*}
\begin{split}
c(v) & \ge \max(\{c^*(u)+w(uv)\ :\ uv\in E,\ u\in X\}\cup\{s^*+[v<\max f^{-1}(1)]\}) \\
       & =   \max(\{s^* +w(uv) - f(u) + 1\ :\ uv\in E,\ u\in X\}\cup\{s^*+[v<\max f^{-1}(1)]\}) \\
       & =  s^* - 1 + \tilde{f}(v).
\end{split}
\end{equation*}
It follows that $|\Dd_{s^*,X,f}|=Q^*[X,f,s^*]\cdot Q[\overline{X},\tilde{f},s-s^*+1]$, as required.
\end{proof}

From Lemma~\ref{lem:compute-Q}, Observation~\ref{obs-Q*} and Lemma~\ref{lem:counting-mim} we infer the following theorem.

\begin{theorem}
For every $\ell$-bounded weight function the number of all proper assignments of a given span can be computed in $O^*(2^n(\ell+1)^{n/2})$ time.
\end{theorem}

\section{Hardness of \TCol}
\label{sec:hardness}

In this section we give a lower bounds for the time complexity of \TCol, under ETH.
To this end we present a reduction from \SetCover. 
The instance of the decision version of \SetCover consists of a family of sets $\mathcal{S}=\{S_1,\ldots,S_m\}$ and a number $k$.
The set $U=\bigcup\mathcal{S}$ is called \emph{the universe} and we denote $n=|U|$.
The goal is to decide whether there is a subfamily $\mathcal{C}\subseteq\mathcal{S}$ of size at most $k$ such that $\bigcup {\cal C} = U$ (then we say the instance is {\em positive}).

In the following lemma we reduce {\sc Set Cover} to the decision version of \TCol, where for a given instance $(G,w)$ and a number $s$ we ask whether there is a proper assignment of  span at most $s$ (then we say the instance is {\em positive}). We say that an instance $(\mathcal{S},k)$ of \SetCover is equivalent to an instance $(G,w,k)$ of \TCol when $(\mathcal{S},k)$ is positive iff $(G,w,k)$ is positive. For every edge $e$ of $G$, every pair $(e,d)$ for $d\in t(e)$ is called a {\em constraint}.

\begin{lemma}
\label{th:SC-to-TC}
Let $(\mathcal{S},k)$ be an instance of \SetCover with $m$ sets
and universe of size $n$ and let $A\in[1,m]$ and $B\in[1,n]$ be two reals.
Then we can generate in polynomial time an equivalent instance of \TCol which has $O\left(\frac{n}{B} + \frac{m}{A} \cdot \max\{1, \log A\}\right)$ vertices, $O^*\left(2^A \cdot m^B\right)$ constraints and is $O\left( 2^A\cdot m^B\right)$-bounded.
%and such that all the numbers in the instance have
%$O\left(A + B\log m + \log n \right)$ bits.
\end{lemma}

\begin{proof}
For convenience we assume that $A$ and $B$ are natural numbers, since otherwise we round $A$ and $B$ down and the whole construction and its analysis is the same, up to some details.

In the proof we  consider coloring of the vertices as placing the vertices on
a number line in such a way that every vertex is placed in the coordinate
equal to its color.

Let $\Ss = \{S_1,\ldots,S_m\}$.
We are going to construct a complex instance $(G=(V,E),t,s)$ of \TCol. 
We describe it step-by-step and show some of its properties.

We begin by putting vertices $v_L$ and $v_R$ in $V$ and $t(v_Lv_R)=\{0,\ldots,s-2\}$, i.e.\ in every proper assignment $v_L$ has color $1$ and $v_R$ has color $s$, or the other way around; w.l.o.g.\ we assume the first possibility. We specify $s$ later.

In what follows, whenever we put a new vertex $v$ in $V$, we will specify the set $A(v)$ of its {\em allowed} colors. 
Formally, this corresponds to putting $t(v_Lv)=\{d\in\{0,\ldots,s-1\}\ :\ d+1\not\in A(v)\}$.

Our instance will consist of three separate modules (the set choice module, the witness module and the parsimonious module).
By separate we mean they have disjoint sets of vertices $V_S$, $V_U$ and $V_P$ and moreover they have disjoint sets of allowed colors, i.e.\ for $i,j\in\{S,U,P\}$, when $x\in V_i$ and $y\in V_j$ for $i\ne j$ then $A(x)\cap A(y) = \emptyset$. However the modules will interfere with each other by forbidding some distances between pairs of vertices from two different modules.

\begin{figure}[t]
\begin{center}
  \begin{minipage}{0.45\textwidth}
    \centering
    \def\svgwidth{2in}
    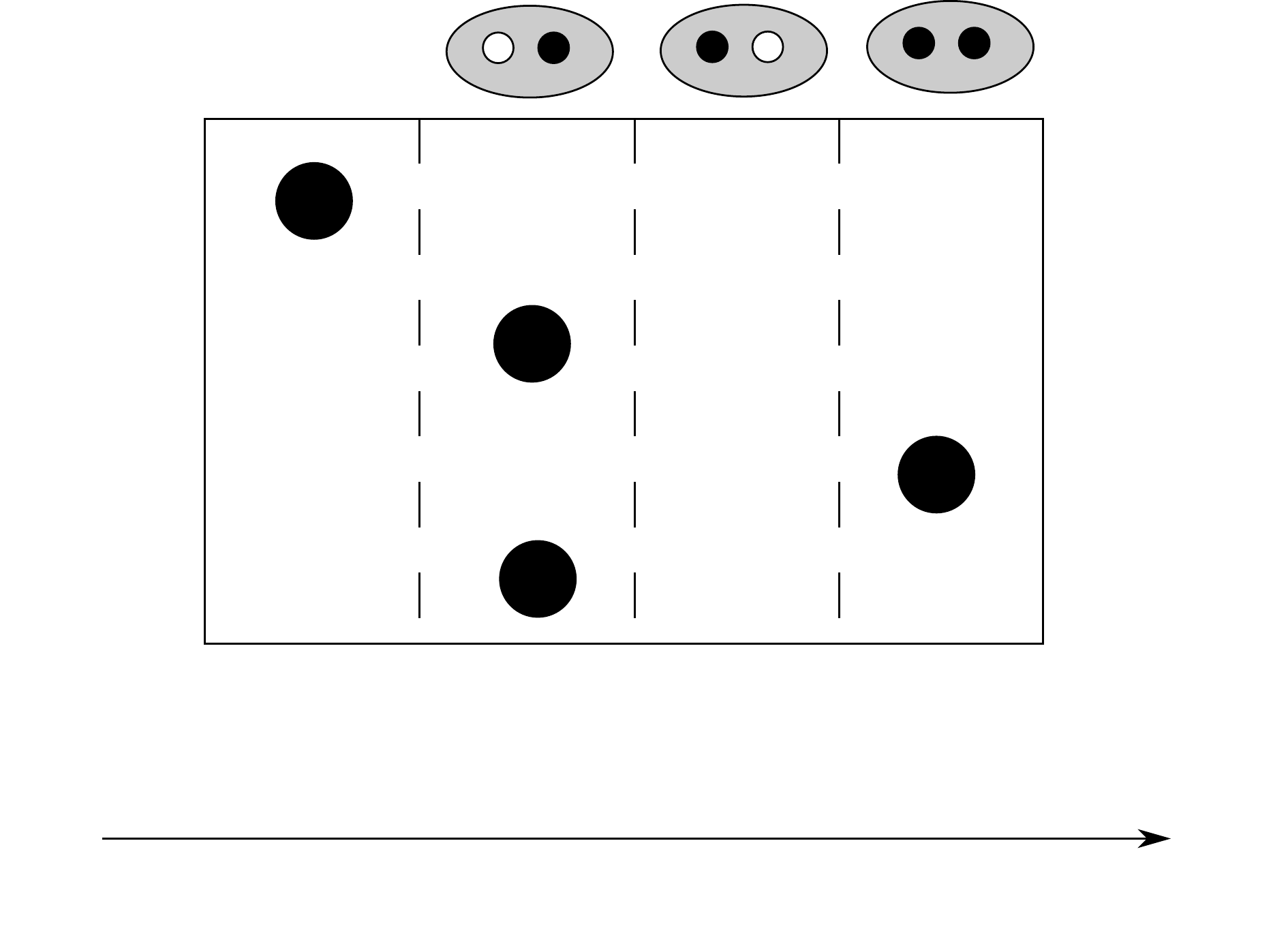
    \caption{The set choice module.}
  \end{minipage}
  \vspace{0.1\textwidth}
  \begin{minipage}{0.45\textwidth}
    \centering
    \def\svgwidth{2in}
    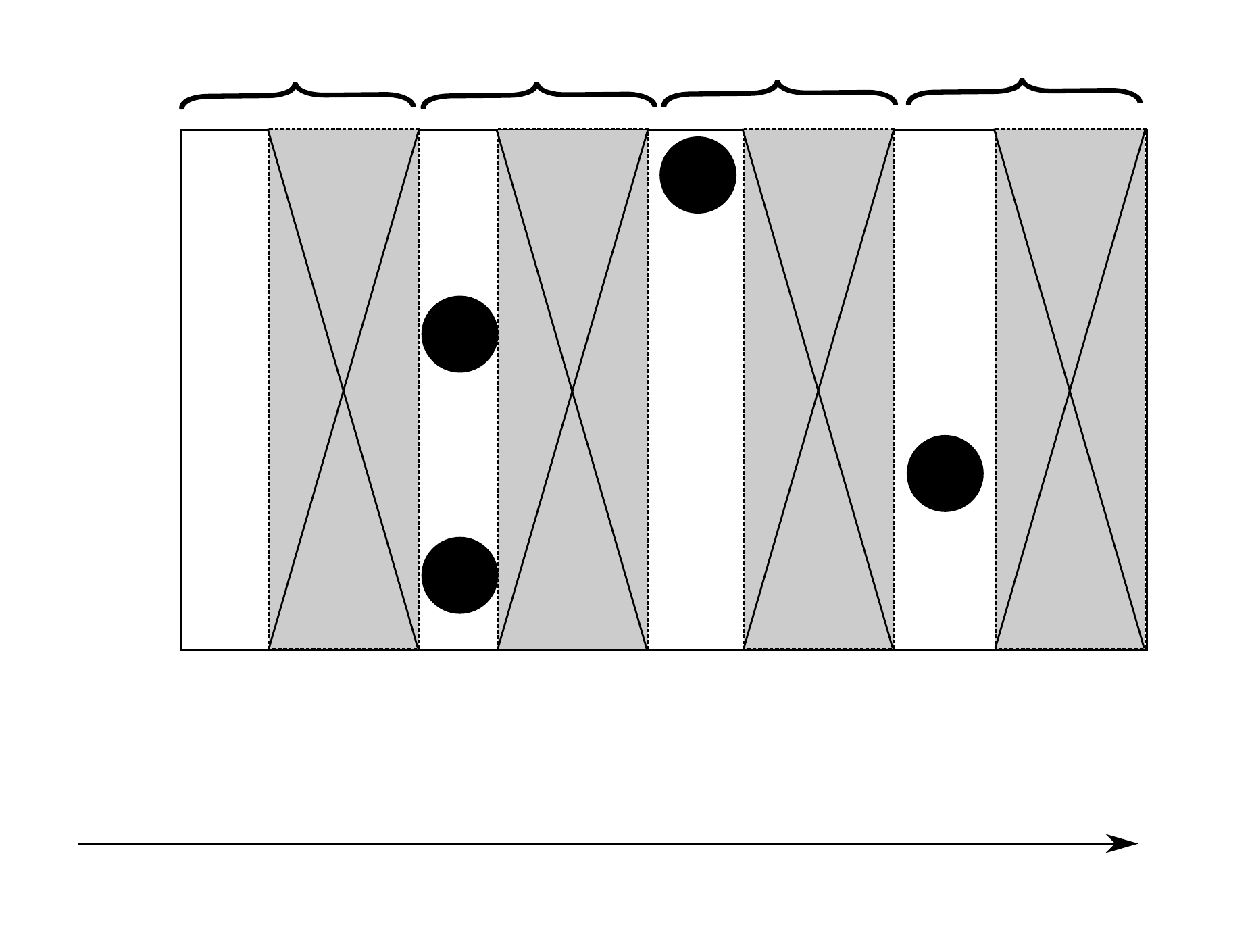
    \caption{
      The witness module. (The grey areas are the gaps between the
      $m^B$ potentially allowed positions.)
    }
  \end{minipage}
\end{center}
\vspace{-2cm}
\end{figure}

\heading{The set choice module}
The first module represents the sets in $\Ss$.
For every $i=1,\ldots,\left\lceil \frac{m}{A} \right\rceil$ the set $V_S$ contains a vertex $s_i$.
Vertex $s_i$ represents the $A$ sets
\[\Ss_i=\{S_{(i - 1) \cdot A + 1}, S_{(i - 2) \cdot A + 2}, \ldots, S_{i \cdot A}\}\] (and the last vertex $s_{\left\lceil {m}/{A} \right\rceil}$ represents $\Ss_{\left\lceil {m}/{A} \right\rceil}=S_{(\left\lceil {m}/{A} \right\rceil-1)A+1}, \ldots, S_m$).
We also put $A(s_i)=\{1,\ldots,2^A\}$ for every $s_i\in V_S$.
The intuition is that the color $c\in[2^A]$ of a vertex $s_i$ corresponds to a subset $\Ss_i(c)\subseteq\Ss_i$, i.e.\ the choice of sets from $\Ss_i$ to the solution of \SetCover.

\heading{The witness module}
Let denote the elements of the universe as $e_1, e_2, \ldots, e_n.$
For every $i=1,\ldots,\left\lceil \frac{n}{B} \right\rceil$ the set $V_U$ contains a vertex $u_i$.
Vertex $u_i$ represents the $B$ elements
\[U_i=\{e_{(i - 1) \cdot B + 1}, e_{(i - 2) \cdot B + 2}, \ldots, e_{i \cdot B}\}\] (and the last vertex $u_{\left\lceil {n}/{B} \right\rceil}$ represents $U_{\left\lceil {n}/{B} \right\rceil}=e_{(\left\lceil {n}/{B} \right\rceil-1)B+1}, \ldots, e_n$).

This time vertices $V_U$ do not need to have the same sets of allowed colors, but for every $u\in V_U$ we have $A(u) \subseteq \{1+i\cdot 2^A\ :\ i=1,\ldots,m^B\}$. 
Note that every vertex has at most $m^B$ allowed colors and there are gaps of length $2^A-1$ where no vertex is going to be assigned.

We say that a sequence $(S_{w_1},\ldots,S_{w_B})\in\Ss^B$ is a {\em witness} for a vertex $u_i\in V_U$ when
\[U_i \subseteq \bigcup_{j=1}^B S_{w_j}.\]
For every $i=1,\ldots,m^B$ color $1+i\cdot 2^A$ corresponds to the $i$-th sequence in the set $\Ss^B$ (say, in the lexicographic order of indices); we denote this sequence by $\Ww_i$. Then, for every $u\in V_U$,
\[A(u) = \{1+i\cdot 2^A\ :\ \text{$\Ww_i$ is a witness for u, }i=1,\ldots,m^B \}.\]
The intuition should be clear: color of a vertex $u_i\in V_U$ in a proper assignment represents the choice of at most $B$ sets in the solution of \SetCover which cover $U_i$.

\heading{The interaction between the set choice module and the witness module}
As we have seen, every assignment $c$ of colors to the vertices determines a choice of a subfamily $\Ss(c)\subseteq\Ss$, where 
$\Ss(c)=\bigcup_{i=1}^{\left\lceil {m}/{A} \right\rceil}\Ss_i(c(i))$. 
Similarly, $c$ determines a choice of a subfamily $\Ss'(c)\subseteq\Ss$, where 
$\Ss'(c)=\bigcup_{u\in V_U}\Ww_{c(u)}$. 
It should be clear that we want to force that in every proper assignment $\Ss'(c)\subseteq \Ss(c)$. To this end we introduce edges between the two modules.

\begin{figure}[t]
  \centering
  \def\svgwidth{3in}
  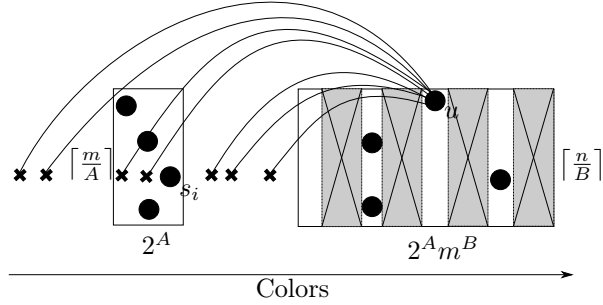
  \caption{
    The interaction between a vertex $s_i$ in the set choice module and a vertex $u$ in the witness module.
    All the drawn arcs are forbidden distances between $s_i$ and $u$.
    Note that for every possible color $1+j\cdot 2^A$ of $u$ the subset of $[2^A]$ excluded
by the forbidden distances in $t(us_i)$ is exactly $F_{i,j}$.
  }
  \vspace{-0.5cm}
\end{figure}

For $i=1,\ldots,\left\lceil \frac{m}{A} \right\rceil$ and $j=1,\ldots,m^B$ define the following set of forbidden colors
\[F_{i,j} = \{c\in [2^A]\ :\ \Ww_j\cap\Ss_i \not\subseteq \Ss_i(c)\}.\]

The intuition is the following: If a proper assignment colors a vertex $u_i\in V_U$ with color $1+j\cdot2^A$ (i.e.\ it assigns the witness $\Ww_j$ to the set $U_i$) then it cannot color the vertex $s_i$ with colors from $F_{i,j}$ (i.e.\ choose this subsets of $\Ss_i$ corresponding to these colors), for otherwise $\Ss'(c)\not\subseteq \Ss(c)$.

{\bf Claim 1} Consider any proper assignment $c:V\rightarrow[s]$. 
If for every $i=1,\ldots,\left\lceil \frac{m}{A} \right\rceil$ we have
$c(s_i)\not\in\bigcup_{u\in V_U} F_{i,c(u)}$, then $\Ss'(c)\subseteq \Ss(c)$.

{\em Proof of the claim:} Consider a set $S_t \in \Ww_{c(u)}$ for an arbitrary $u\in V_U$. Then $S_t\in \Ss_i$ for some $i$. From the assumption, $c(s_i)\not\in F_{i,c(u)}$, so $\Ww_{c(u)}\cap \Ss_i \subseteq \Ss_i(c)$. Hence, $S_t \in \Ss(c)$, as required.

Hence we would like to add some forbidden distances to our instance to make the assumption of Clam 1 hold. To this end, for every $u\in V_U$ and every $s_i\in V_S$ we put
\[t(us_i) = \bigcup_{j=1}^{m^B}\{1+j\cdot 2^A-f\ :\ f\in F_{i,j}\}.\]

In other words, for every possible color $1+j\cdot 2^A$ of $u$ we forbid all distances between $u$ and $s_i$ that would result in coloring $s_i$ with $F_{i,j}$. Then indeed the assumption from Claim 1 holds.

{\bf Claim 2} For any proper assignment $c:V\rightarrow[s]$ we have  $\Ss'(c)\subseteq \Ss(c)$.

{\em Proof of the claim:} We need to verify the assumption in Claim 1. Assume for the contradiction that for some $i$ and some $u\in V_U$ we have $c(s_i)\in F_{i,c(u)}$.
Recall that in a proper assignment $c(u)=1+j\cdot 2^A$ for some $j=1,\ldots,m^B$.
Then $|c(u)-c(s_i)|=1+j\cdot 2^A-c(s_i)\in t(us_i)$, a contradiction.

{\bf Claim 3} For any proper assignment $c:V\rightarrow[s]$ we have  $\Ss(c)$ covers the universe.

{\em Proof of the claim:} This is an immediate corollary from Claim 2 and the fact that every vertex $u\in V_U$ is colored with a color from $A(u)$.

{\bf Claim 4} For every cover $\Cc\subseteq\Ss$ of the universe, there is a proper assignment $c:V\rightarrow[s]$ such that $\Ss(c)=\Cc$.

{\em Proof of the claim:} We color $v_L$ and $v_R$ with $1$ and $s$, and every vertex $s_i$ with the color from $[2^A]$ corresponding to the subset $\Ss_i\cap \Cc$ of $\Ss_i$. 
For every set $U_i$ for every $e\in U_i$ we pick a set $S_e\in \Cc$ that contains $e$ and we build a witness $\Ww$ from the sets $S_e$. We color $u_i$ with the color $1+j\cdot 2^A$, where $j$ is the number of $\Ww$ in the lexicographic order of all witnesses. It remains to check that the resulting assignment $c$ is proper. The only nontrivial issue is whether for every $u\in V_U$ and $s_i\in V_S$ we have $|c(u)-c(s_i)|\not\in t(us_i)$. It is clear that $|c(u)-c(s_i)|\not\in\{c(u)-f\ :\ f\in F_{i,j}\}$, where $j$ is such that $c(u)=1+j\cdot 2^A$. However, for every $j'\ne j$ the set $\{1+j'\cdot 2^A-f\ :\ f\in F_{i,j'}\}$ is disjoint from $2^A$ (this is where we make use of the `gaps' of length $2^A-1$). 

\heading{Bounding the number of sets chosen to the solution}
The last thing we need in a proper assignment $c$ is to keep the number of the sets in $\Ss(c)$ bounded by $k$. To  this end we use the parsimonious module with the vertex set $V_P$.

The third parsimonious module consists of $\left\lceil \frac{m}{A} \right\rceil$
consecutive submodules and an additional free space of length $2k$ (meaning that for every $v\in V_P$ the set of allowed colors $A(v)$ contains this free space.
Between those submodules and the additional free space we put a gap of length
$2^A$, where no vertex can be assigned.
The intuition is that in a proper assignment $c$ the $i$-th submodule represents the number of sets from $\Ss_i$ chosen to the solution, i.e.\ $|\Ss_i(c_i)|$.

\begin{figure}[t]
  \centering
  \def\svgwidth{3in}
  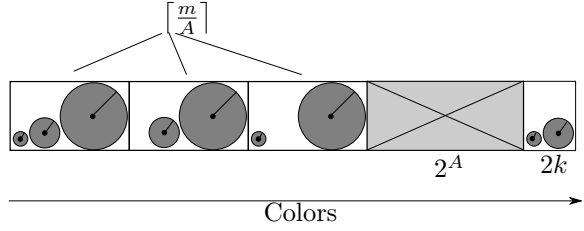
  \caption{The parsimonious module.\label{fig:disks}}
  \vspace{-0.5cm}
\end{figure}

More precisely, $V_P=\biguplus_{i=1}^{\ceil{m/A}} V_i$, where $V_i$ is a set of $1 + \left\lfloor \log A \right\rfloor$ vertices
representing numbers $2^0, 2^1, \ldots, 2^{\left\lfloor \log A \right\rfloor}.$
Let For a vertex $x\in V_P$ let $r(x)$ denote the number represented by $x$.
For every two vertices $x,y\in V_P$ we define 
\[t(xy)=\{0,\ldots,r(x)+r(y)-1\}.\]
It follows that we can interpret those vertices as disjoint disks with radii equal to
the represented numbers (see Fig.~\ref{fig:disks}).
Let $q=(1+m^B)2^A$, i.e.\ $q$ is the number of colors used by the first two modules. 
For every $i$, we define $i$-th slot as the set of colors $\{q+1+(i-1)\cdot 4A,\ldots,q+i\cdot 4A\}$. Note that the length of each slot is $4A$. Define also the free space as $Q=\{q+\ceil{m/A}\cdot 4A+2^A+1,\ldots,q+\ceil{m/A}\cdot 4A+2^A+2k\}$. 
Each vertex $x\in V_i$ is either in $i$-th slot or in the free space $Q$.
However, $x$ has exactly one allowed color in the $i$-th slot chosen so that we can put all the disks in the $i$-th slot and they will be disjoint. 
Let $j$ be such that $r(x)=2^j$. Then we denote the allowed color by $a_x=q+(i-1)\cdot 4A+\sum_{r<j}2\cdot 2^r+2^j$.
In precise terms, $A(x)=\{a_x\}\cup Q$.

\begin{figure}[h!]
  \centering
  \def\svgwidth{3in}
  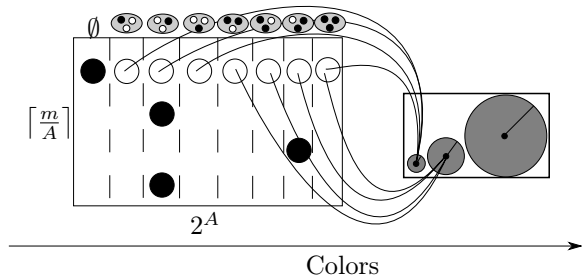
  \caption{
    The interaction between the set choice module and one of
    the submodules of the parsimonious module.
    Note that colors in $[2^A]$ are ordered according to the cardinality of the chosen
    collection of sets ($0, 1, 1, 1, 2, 2, 2, 3$).  
 %    Note that the submodule represents the counter counting not the position
%    of the corresponding vertex in the first part, but the number of chosen sets
%    represented by that position (here -- for the purpose of the illustration
%    -- the following positions of the corresponding vertex mean that this
%    vertex chooses respectively $0, 1, 1, 1, 2, 2, 2$ and $3$ sets to the
%    solution).
  }
  \vspace{-0.5cm}
\end{figure}

Vertices of the $i$-th submodule have some edges to the vertex $s_i$ of the
set choice module. As we mentioned, for a proper assignment $c$ the $i$-th submodule is going to be a counter representing the number of sets in $\Ss_i(c)$; in fact the vertex representing $2^j$ corresponds to the $j$-th bit of the counter.
So if $r(x)=2^j$ for $x\in V_i$, then $t(s_ix)$ contains all distances $d$ such that
$a_x-d$ is a color $b$ from $2^A$ such that the $j$-th bit of $|\Ss_i(b)|$ is 1.
Hence, in a proper assignment $c$, if the $j$-th bit of the number of sets in $\Ss_i(c)$ is 1 then $x$ is thrown away from the $i$-th slot and it is colored by a color from the free space $Q$.
However, the $|Q|=2k$ so the sum of the radii of the disks thrown out from its slots is at most $k.$ It follows that the total number of the chosen sets is also at most $k.$
Also, if there is a cover $\Cc\subseteq\Ss$ of the universe such that $|\Cc|\le k$, then for every $i$, if $|\Cc\cap\Ss_i|$ has 1 on the $j$-th bit we put the vertex of $V_i$ representing $2^j$ in $Q$. It is clear that since $|\Cc|\le k$ we have enough space for them in $Q$. Moreover, we do not violate any edge between these vertices and $V_S$ because of the gap $2^A$ inside the parsimonious module. Together with Claim 4 it implies that $(\Ss,k)$ is a YES-instance of \SetCover iff $(G,t,s)$ is a YES-instance of \TCol, provided that $s$ is sufficiently large to provide disjoint intervals of colors for all the modules. From the construction we infer that it is sufficient to put $s = 2^A + 2^A \cdot m^B + 4A \cdot \left\lceil\frac{m}{A}\right\rceil + 2^A + 2k.$

\heading{Calculating the parameters}
Note that $s=O(2^Am^B)$ and in particular our instance is $O(2^Am^B)$-bounded.
Moreover, $|V|=
  \left\lceil\frac{n}{B}\right\rceil +
  \left\lceil\frac{m}{A}\right\rceil +
  \left\lceil\frac{m}{A}\right\rceil \cdot
    \left(1 + \left\lfloor\log A\right\rfloor \right)
  + 2
= O\left(\frac{n}{B} + \frac{m}{A} \cdot \left(1 + \log A\right) \right)
= O\left(\frac{n}{B} + \frac{m}{A} \cdot \max\left\{1, \log A\right\} \right)$.
Finally, the total number of constraints is bounded by
$
O^*(
 (\frac{n}{B} + \frac{m}{A} \cdot \max\left\{1, \log A\right\} )^2
 \cdot (2^A \cdot m^B)
 )
= O^*\left(2^A \cdot m^B\right)$, i.e., 
the number of pairs of the vertices times the maximum forbidden distance $s-1$.
It ends the proof.
\ignore{
To denote every forbidden distance as well as the searched span in the query
it is sufficient to use  $
\log\left(\left(( 2^A\cdot m^B\right)\cdot {\rm poly}(A, B, m,n) \right)
= A + B\log m + \log\left({\rm poly}(A, B, m, n)\right)
= O\left(A + B\log m + \log n\right) $
bits.
To denote the number of the vertices in the instance as well as to denote
the numbers of specific vertices it is sufficient to use
$O\left(\log(n+m)\right)$ bits.
So clearly to denote every number in our instance of \TCol
it suffices to use $O\left(A + B\log m + \log n\right)$ bits.

The last thing is to deal with the real values of $A$ and $B.$
We can perform all the above construction for integer numbers $\lfloor A\rfloor$
and $\lfloor B\rfloor$ instead of the real $A$ and $B.$
Then we obtain an instance with $
O\left(
  \frac{n}{\lfloor B\rfloor} +
  \frac{m}{\lfloor A\rfloor} \cdot \max\left\{1, \log \lfloor A\rfloor\right\}
  \right)
= O\left(
  \frac{n}{B} +
  \frac{m}{A} \cdot \max\left\{1, \log A\right\}
  \right)
$ vertices (because if $x \geq 1$ then
$\frac{1}{\lfloor x\rfloor} \leq \frac{2}{x}$),
with the overall number of the forbidden distances $
O^*\left(2^{\lfloor A\rfloor} \cdot m^{\lfloor B\rfloor}\right)
= O^*\left(2^A \cdot m^B\right)$
and with all the numbers in the instance with
$O\left(\lfloor A\rfloor + \lfloor B\rfloor\log m + \log n\right)
= O\left(A + B\log m + \log n\right)$
bits. 
}
\end{proof}

\begin{corollary}
\label{cor:DS-to-TC}
Let $(G, k)$ be an instance of \DomSet where $G$
is a graph on $n$ vertices and $k\in\mathbb{N}$.
Then, for any real number $A \in [1,n]$
we can generate in polynomial time an instance of \TCol with
$O\left(\frac{n}{A} \cdot \max\{1, \log A\}\right)$ vertices
and with $O^*\left((2n)^A\right)$ constraints
and such that all the numbers in the instance have
$O\left(A \cdot \max\left\{1, \log n\right\}\right)$ bits.
\end{corollary}

\begin{proof}
The instance of \textsc{Dominating Set} with $n$ vertices can be transformed to an equivalent instance of \textsc{Set Cover} with $n$ sets and also $n$ elements of
the universe in a standard way (the sets are exactly the neighborhoods of the vertices).
The number $k$ stays the same.
Therefore we can use the Lemma~\ref{th:SC-to-TC} with $A=B$ and $m=n.$
\end{proof}

\begin{theorem}
\label{th:DS-to-TC-complex}
If there exists an algorithm solving \TCol
in one of the following time complexities:
\begin{enumerate}[$(i)$]
  \item $2^{2^{o\left(\sqrt{n}\right)}} {\rm poly}(r),$
  \item $2^{n \cdot o\left({\log l}/{(\log\log l)^2}\right)} {\rm poly}(r),$
\end{enumerate}
where $n$ is the number of vertices in the input graph and $r$ is the bit size of the input, then there exists an algorithm solving \textsc{Dominating Set}
in time $2^{o(n)}.$
\end{theorem}

\begin{proof} We begin with proving $(i)$.
Let us assume that we have an algorithm solving \TCol
    in time $2^{2^{f(n)}} {\rm poly}(r)$
    where $f$ is some function such that $f(n) = o\left(\sqrt{n}\right).$
    We can assume without loss of generality that $f$ is positive and
    nondecreasing.
    Let $C$ be a constant such that Corollary \ref{cor:DS-to-TC} will give us
    always at most $C \cdot \frac{n}{A} \cdot \max\left\{1, \log A\right\}$
    vertices.
    Let $\alpha$ be a positive nondecreasing function such that
    $\alpha(n) \leq \frac{\sqrt{n}}{f(Cn)}$ and $\alpha(n) = \omega(1).$
    Such a function always exists because
    $\frac{\sqrt{n}}{f(Cn)} = \frac{1}{\sqrt{C}} \cdot \frac{\sqrt{Cn}}{f(Cn)}
      = \omega(1).$
    For every instance of \DomSet with $n$ vertices we can
    take $A = \frac{n}{\alpha\left(\log^2 n\right) \log n}$ and use
    Corollary \ref{cor:DS-to-TC} to obtain an instance of \TCol
    with $O\left(\frac{n}{A} \log A\right)
      = O\left(\alpha\left(\log^2 n\right) \log^2 n\right)$
    vertices and
    \[O^*\left((2n)^A\right) = O^*\left(2^{A + A\log n}\right)
      = O^*\left(2^{O\left({n}/{\left(\alpha\log^2 n\right)}\right)}\right)
      = 2^{o(n)}\]
    constraints. Moreover the numbers in the instance have polynomial size,
    so the size of the whole instance is $2^{o(n)}.$
    Thus this instance can be built in
    ${\rm poly}\left(n, 2^{o(n)}\right) = 2^{o(n)}$ time.
    Then we can solve this instance in
    $2^{2^{f\left(C \cdot \frac{n}{A}\log A\right)}}
      {\rm poly}\left(2^{o(n)}\right)$
    time. But
    $f\left(C \cdot \frac{n}{A} \log A\right)
      \leq f\left(C \cdot \alpha\left(\log^2 n\right) \log^2 n\right)
      \leq \frac{\sqrt{\alpha\left(\log^2 n\right) \log^2 n}}{
          \alpha\left(\alpha\left(\log^2 n\right) \log^2 n\right)}
      = $ \linebreak $\log n \cdot \frac{\sqrt{\alpha\left(\log^2 n\right)}}{
          \alpha\left(\alpha\left(\log^2 n\right) \log^2 n\right)}
      = \log n \cdot \frac{\sqrt{\alpha\left(\log^2 n\right)}}{
            \alpha\left(\log^2 n\right)}
          \cdot \frac{\alpha\left(\log^2 n\right)}{
             \alpha\left(\alpha\left(\log^2 n\right) \log^2 n\right)}
      \leq \frac{\log n}{\sqrt{\alpha\left(\log^2 n\right)}}
      = o\left(\log n \right)$.
    %for sufficiently big values of $n,$ because then we have
    %$\alpha\left(\log^2 n\right) \geq 1$ and then
    %$\alpha\left(\log^2 n\right)
    %  \leq \alpha\left(\alpha\left(\log^2 n\right)\log^2 n\right).$
    So the time of the whole procedure is
    $2^{o(n)} + 2^{2^{o\left(\log n\right)}}{\rm poly}\left(2^{o(n)}\right)
      = 2^{o(n)}.$

  Now we focus on $(ii)$. Let us assume we have an algorithm solving \TCol in time
    $2^{n \cdot f(l)} {\rm poly}(m)$ where $f$ is a positive function such that
    $f(l) = o\left(\frac{\log l}{\log^2\log l}\right).$
    Let $A = \frac{n}{\log^2 n}.$
    For every instance of \DomSet we can use
    Corollary \ref{cor:DS-to-TC} to obtain an instance of
    \TCol with
    $O\left(\log^2 n \cdot \log \frac{n}{\log^2 n}\right)
      = O\left(\log^3 n\right)$ vertices,
    $O^*\left((2n)^\frac{n}{\log^2 n}\right) 
      = 2^{O\left(\frac{n}{\log n}\right)}
      = 2^{o(n)}$ constraints and
    every number with $O\left(\frac{n}{\log n}\right)$ bits.
    We can obtain it in ${\rm poly}\left(n, 2^{o(n)}\right) = 2^{o(n)}$ time. 
    Note that then $\log l \leq C \frac{n}{\log n}$ for some constant $C.$
    The function ${x}/{\log^2 x}$ is nondecreasing for big values of
    $x$ so for big values of $n$ we have ${\log l}/{\log^2\log l}
      \leq {C \frac{n}{\log n}}/{\log^2 \left(C \frac{n}{\log n}\right)}.$
    So we can solve our instance of \TCol in time
    \begin{equation*}
    \begin{split}
& 2^{O\left(\log^3 n\right) \cdot
      o\left({C \frac{n}{\log n}}/{\log^2\left(C\frac{n}{\log n}\right)}
        \right)} \cdot {\rm poly}\left(2^{o(n)}\right)
     = 2^{o\left({n\log^2 n}/{\log^2\left(C\frac{n}{\log n}\right)}\right)}
       \cdot 2^{o(n)}
     = \\ 
     & 2^{o\left({n\log^2 n}/{\left(\log C + \log n - \log\log n\right)^2}
       \right)} \cdot 2^{o(n)} 
      = 2^{o\left({n}/{\left(\frac{\log C}{\log n} + 1
       - \frac{\log\log n}{\log n}\right)^2}\right)} \cdot 2^{o(n)} = \\
     & 2^{o(n)} \cdot 2^{o(n)}
     = 2^{o(n)}.     
    \end{split}
    \end{equation*}
     So we have solved the given instance of \textsc{Dominating Set} in time
     $2^{o(n)} + 2^{o(n)} = 2^{o(n)}.$
      
\end{proof}

\begin{corollary}
\label{cor:TC-ETH}
There is no algorithm solving an $n$-vertex instance of \TCol
with bit size $r$ in any of the listed time complexities
\begin{itemize}
  \item $2^{2^{o\left(\sqrt{n}\right)}} {\rm poly}(r),$ 
  \item $2^{n \cdot o\left({\log l}/{\log^2\log l}\right)} {\rm poly}(r),$
 % \item $2^{o(n) \cdot \frac{\log l}{\log^2\log l}} {\rm poly}(m).$
\end{itemize}
unless the Exponential Time Hypothesis fails.
\end{corollary}

\begin{proof}
Under the ETH assumption there is no algorithm solving \textsc{Dominating Set}
in time $2^{o(n)}$ where $n$ is a number of the vertices (See~\cite{FominKW04}).
Therefore the claim follows immediately from Theorem \ref{th:DS-to-TC-complex}.
\end{proof}

Regarding the first claim the theorem above, we note that there is a $2^{O(n\log l)} {\rm poly}(r)$-time algorithm for \TCol, see~\cite{junosza-tcol}.

\bibliographystyle{abbrv}
\bibliography{channel}

\begin{thebibliography}{10}

\bibitem{bhk:coloring}
A.~Bj{\"o}rklund, T.~Husfeldt, and M.~Koivisto.
\newblock Set partitioning via inclusion-exclusion.
\newblock {\em SIAM J. Comput.}, 39(2):546--563, 2009.

\bibitem{cygan}
M.~Cygan and L.~Kowalik.
\newblock Channel assignment via fast zeta transform.
\newblock {\em Inf. Process. Lett.}, 111(15):727--730, 2011.

\bibitem{FominKW04}
F.~V. Fomin, D.~Kratsch, and G.~J. Woeginger.
\newblock Exact (exponential) algorithms for the dominating set problem.
\newblock In {\em Proc. WG'04}, volume 3353 of {\em Lecture Notes in Computer
  Science}, pages 245--256, 2004.

\bibitem{Hale80}
W.~Hale.
\newblock Frequency assignment: Theory and applications.
\newblock {\em Proceedings of the IEEE}, 68(12):1497--1514, Dec 1980.

\bibitem{HS-JACM74}
E.~Horowitz and S.~Sahni.
\newblock Computing partitions with applications to the knapsack problem.
\newblock {\em J. ACM}, 21(2):277--292, 1974.

\bibitem{dagstuhl}
T.~Husfeldt, R.~Paturi, G.~B. Sorkin, and R.~Williams.
\newblock {Exponential Algorithms: Algorithms and Complexity Beyond Polynomial
  Time (Dagstuhl Seminar 13331)}.
\newblock {\em Dagstuhl Reports}, 3(8):40--72, 2013.

\bibitem{IP01}
R.~Impagliazzo and R.~Paturi.
\newblock On the complexity of k-sat.
\newblock {\em J. Comput. Syst. Sci.}, 62(2):367--375, 2001.

\bibitem{junosza-tcol}
K.~Junosza-Szaniawski and P.~Rzążewski.
\newblock An exact algorithm for the generalized list {$T$}-coloring problem.
\newblock {\em CoRR}, abs/1311.0603, 2013.

\bibitem{kral}
D.~Kr{\'a}l.
\newblock An exact algorithm for the channel assignment problem.
\newblock {\em Discrete Applied Mathematics}, 145(2):326--331, 2005.

\bibitem{mcdiarmid}
C.~J.~H. McDiarmid.
\newblock On the span in channel assignment problems: bounds, computing and
  counting.
\newblock {\em Discrete Mathematics}, 266(1-3):387--397, 2003.

\bibitem{traxler}
P.~Traxler.
\newblock The time complexity of constraint satisfaction.
\newblock In M.~Grohe and R.~Niedermeier, editors, {\em IWPEC}, volume 5018 of
  {\em Lecture Notes in Computer Science}, pages 190--201. Springer, 2008.

\end{thebibliography}

\end{document}